\documentclass[final,3p,times]{elsarticle}

\usepackage{hyperref}
\usepackage{amsthm,amsmath,amssymb,eucal}
\usepackage{caption}
\usepackage{subcaption}

\journal{Physics Letters A}

\newtheorem{theorem}{Theorem}

\newdefinition{remark}{Remark}

\begin{document}

\begin{frontmatter}

\title{Topological bulk-edge effects in quantum graph transport}





\author[label1,label2]
{Pavel Exner}
\ead{exner@ujf.cas.cz}
\author[label3]
{Ji\v{r}\'{\i} Lipovsk\'{y}\corref{cor1}}
\ead{jiri.lipovsky@uhk.cz}
\address[label1]{Doppler Institute for Mathematical Physics and Applied Mathematics, Czech Technical University,
B\v rehov{\'a} 7, 11519 Prague, Czechia}
\address[label2]{Department of Theoretical Physics, Nuclear Physics Institute, Czech Academy of Sciences, 25068 \v{R}e\v{z} near Prague, Czechia}
\address[label3]{Department of Physics, Faculty of Science, University of Hradec Kr\'alov\'e, Rokitansk\'eho 62, 500 03 Hradec Kr\'alov\'e, Czechia}

\begin{abstract}
We examine quantum transport in periodic quantum graphs with a vertex coupling non-invariant with respect to time reversal. It is shown that the graph topology may play a decisive role in the conductivity properties illustrating this claim with two examples. In the first, the transport is possible at high energies in the bulk only being suppressed at the sample edges, while in the second one the situation is opposite, the transport is possible at the edge only.
\end{abstract}


\end{frontmatter}


\section{Introduction}

Quantum graphs, proposed as a single purpose model in the early days of quantum mechanics \cite{Pa} and then forgotten for decades to be rediscovered in the late 1980s, proved to be an exceptionally fruitful concept; for its history, presentation, and an extensive bibliography we refer to the monograph \cite{BK}. One of the sources of its versatility, to our opinion not fully explored, lays in the possibility to choose from a wide variety of conditions coupling the wavefunctions at the graph vertices. Most often one uses the simplest option, usually, called Kirchhoff coupling, or its simple generalization, but there are other choices which could be of a considerable interest.

One concerns the coupling which could be dubbed a \emph{preferred orientation} one, proposed recently by one of us together with M.~Tater \cite{ETa} in response to an attempt to employ quantum graphs to model the anomalous Hall \cite{SK}. The characteristic property of this coupling is that it \emph{violates the time reversal symmetry}. This non-invariance is most pronounced at a particular energy value (referring to momentum $k=1$, the scale being chosen arbitrarily) when the particles approaching the vertex move to the neighboring edges cyclically. The most surprising property of this coupling bears a topological character: transport properties of such a vertex at high energies are determined by its \emph{degree parity}. This can be seen, in particular, from spectral analysis of various classes of periodic graphs \cite{ETa, BET} as well of finite graphs exhibiting a high symmetry \cite{EL}.

Our aim here is to present another manifestation of this property, this time inspecting graphs which are locally periodic but have a boundary where the vertex parity may differ from that of the graph `bulk'. We analyze two simple graph classes with the opposite properties: in the first case the parity is odd at the graph boundary and even `inside', and \emph{vice versa} in the second model. In the first case transport is possible in the bulk but the wavefunctions are suppressed -- except for particular narrow energy intervals -- at the boundary.

In the second case, on the contrary, the motion is possible at the boundary and the wavefunctions decay with respect to the distance from it. We stress that despite the similarity, this system is \emph{not} a topological insulator in the usual sense \cite{HK} because the corresponding `edge states' can move in \emph{both directions} along the array of edges representing the graph boundary -- cf.~Remark~(c) at the end of Section~\ref{s:brick}.

\section{Description of the model}

To begin with, let us recall briefly a few key facts about quantum graphs. We consider a metric graph with infinitely many edges identified with finite intervals $(0,\ell_j)$ and equip it with a second-order operator acting as the negative second derivative. This will be our model Hamiltonian with the common convention used when the physical constants values are not important. To make it a self-adjoint operator one has to fix the \emph{vertex coupling:} the definition domain will consist of functions with their edge components being elements of the Sobolev space $W^{2,2}(0,\ell_j)$ and satisfying the matching condition $(U_{v}-I)\Psi_v+ i(U_{v}+I)\Psi_v' = 0$ at the vertex $v$, where $i$ is the complex unit, $U_{v}$ is a unitary square $d_v\times d_v$ matrix where $d_v$ is the degree of the vertex $v$, $\:I$ is the $d_v\times d_v$ identity matrix, and finally, $\Psi_v$ is the vector of the function values at the vertex understood as one-sided limits when the vertex is approached from a particular edge, and $\Psi_v'$ is similarly the vector of the (one-sided, outward) derivatives at this vertex. For more details on quantum graphs we refer to the monograph \cite{BK}.

The choice of the vertex coupling leaves a lot of freedom, for a vertex of degree $d_v$ we have a $d_v^2$-parameter family of them and, in general, every one defines a different physics. In this paper we consider the coupling with a \emph{preferred orientation} introduced in \cite {ETa} mentioned in the introduction. The matrix appearing in the matching conditions is in this case
\begin{equation}\label{matU}
  U_v = \begin{pmatrix}0 & 1 & 0 & 0 & \cdots & 0 & 0\\ 0 & 0 & 1 & 0 & \cdots & 0 & 0\\ 0 & 0 & 0 & 1 & \cdots & 0 & 0\\ \vdots & \vdots & \vdots & \vdots & \ddots & \vdots & \vdots\\ 0 & 0 & 0 & 0 & \cdots & 0 & 1\\ 1 & 0 & 0 & 0 & \cdots & 0 & 0\\ \end{pmatrix}\,.
\end{equation}
For a particular energy value referring here to the momentum $k = 1$ the wave approaching the vertex over a given edge is fully transmitted to the neighbouring one, conventionally in the counter-clockwise sense. The most remarkable property of this coupling is the dependence of its high-energy transport properties on the topology, specifically on the vertex degree parity as we have recalled above.

As in \cite{ETa, BET} we are here concerned with the transport properties of infinite periodic graphs, this time in the form of a strip cut from an infinite two-dimensional lattice. Our main observation is that the vertex degrees at the boundary of such a strip typically differ from those of the original lattice which in combination with the above mentioned asymptotic behavior may make the transport different in the `bulk' of such a sample and at its `edge'.

We are going to illustrate this claim on two types of lattices from which we cut in both cases a `vertical' strip. In the first case we start from a rectangular lattice of the edges lengths $\ell_j$, $j=1,2$, with the odd degree of the vertices at the strip edges while those inside the strip are of an even degree. We note that the corresponding operator spectrum may have a point component if $\ell_1/\ell_2\in\mathbb{Q}$ \cite[Sec.~4.5]{BK}, however, away from the values $(\pi/\ell_j)^2$, $j=1,2$ and the value one, the spectrum is absolutely continuous; we are going to show that at high energies the corresponding generalized eigenfunctions are suppressed at the vicinity of the edges.

On the other hand, we construct a different type strip cut from a `brick' lattice with the edge lengths $\ell_j$, $j=1,2,3$, in which only the vertices at the left edge are of an even degree, while the degree of all the other vertices is odd. Staying again away from the energy values which might give rise to eigenvalues, we find that at high energies the generalized eigenfunctions are supported in the vicinity of the left edge.

\section{Rectangular lattice}

\begin{figure}
\centering
\includegraphics[height=7cm]{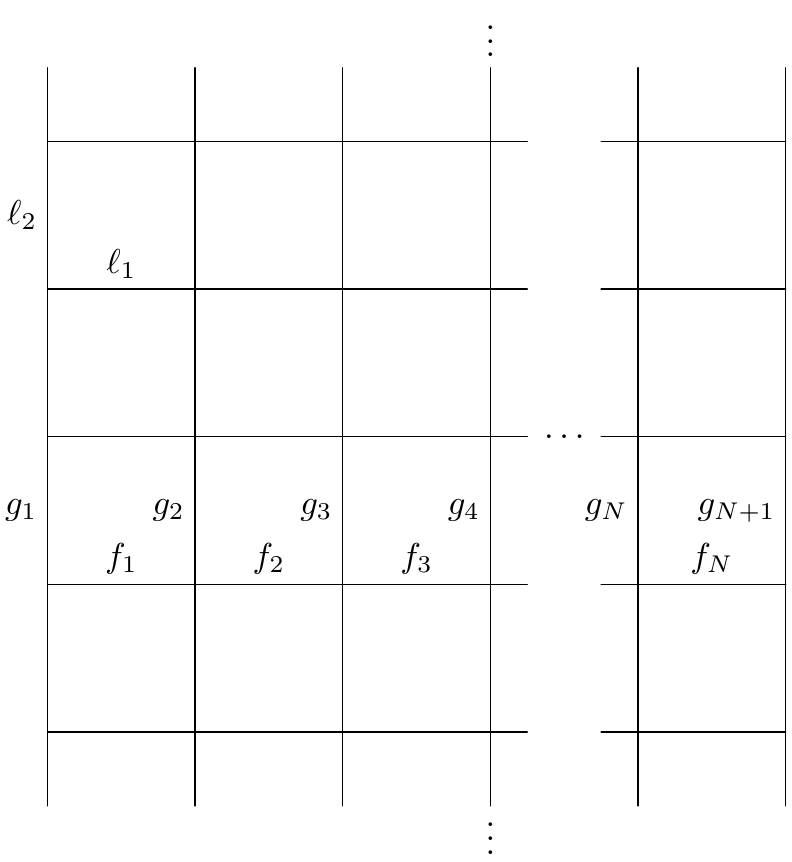}
\caption{Rectangular lattice}
\label{fig1}
\end{figure}

Consider first a two-dimensional rectangular lattice sketched in Figure~\ref{fig1}. The lengths of the horizontal edges are $\ell_1$, the lengths of the vertical edges are $\ell_2$; we cut from the lattice a vertical strip which contains $N$ cells in the horizontal direction. As we have said above, we are interested in the absolutely continuous spectrum of the system. It can be found using the Floquet-Bloch analysis; the elementary cell consists of $N$ horizontal edges in one horizontal line -- the wavefunction components on them are denoted by $f_j$, $j=1,\dots, N$, the edges being identified with the intervals $(0,\ell_1)$ where $x=0$ corresponds to the left vertex of the edge and $x=\ell_1$ to the right one -- and $N+1$ vertical edges with the `lower' vertex, corresponding to $x=0$, common with a vertex of the previously mentioned set of edges, and the other end, $x=\ell_2$, referring to the `neighboring upper floor' of the graph; the wavefunction components on these edges are denoted as $g_j$, $j=1,\dots, N+1$. The elementary cell just above (below) the described one has wavefunction components $\mathrm{e}^{\pm i\theta} f_j$ and $\mathrm{e}^{\pm i\theta} g_j$, respectively, etc.

By assumption, the coupling conditions at the vertices are
\begin{eqnarray*}
  \begin{pmatrix}-1 & 1 & 0\\ 0 &-1 & 1\\ 1 & 0 & -1\end{pmatrix}\begin{pmatrix}g_1(0)\\ \mathrm{e}^{-i\theta}g_1(\ell_2)\\f_1(0) \end{pmatrix}+ i\begin{pmatrix}1 & 1 & 0\\ 0 &1 & 1\\ 1 & 0 & 1\end{pmatrix}\begin{pmatrix}g_1'(0)\\ -\mathrm{e}^{-i\theta}g_1'(\ell_2)\\f_1'(0) \end{pmatrix} &=& 0\,,\\
  \begin{pmatrix}-1 & 1 & 0 & 0\\ 0 &-1 & 1 & 0\\ 0 & 0 & -1 & 1\\ 1 & 0 & 0 & -1\end{pmatrix}\begin{pmatrix}g_j(0)\\ f_{j-1}(\ell_1)\\ \mathrm{e}^{-i\theta}g_j(\ell_2)\\f_j(0) \end{pmatrix} + i \begin{pmatrix}1 & 1 & 0 & 0\\ 0 &1 & 1 & 0\\ 0 & 0 & 1 & 1\\ 1 & 0 & 0 & 1\end{pmatrix} \begin{pmatrix}g_j'(0)\\- f_{j-1}'(\ell_1)\\ -\mathrm{e}^{-i\theta}g_j'(\ell_2)\\f_j'(0) \end{pmatrix} & = & 0\,,\quad j=2, \dots, N\,,
\\
  \begin{pmatrix}-1 & 1 & 0\\ 0 &-1 & 1\\ 1 & 0 & -1\end{pmatrix}\begin{pmatrix}g_{N+1}(0)\\ f_{N}(\ell_1)\\ \mathrm{e}^{-i\theta}g_{N+1}(\ell_2) \end{pmatrix}+ i\begin{pmatrix}1 & 1 & 0\\ 0 &1 & 1\\ 1 & 0 & 1\end{pmatrix}\begin{pmatrix}g_{N+1}'(0)\\ -f_{N}'(\ell_1)\\ -\mathrm{e}^{-i\theta}g_{N+1}'(\ell_2) \end{pmatrix} &=& 0\,.
\end{eqnarray*}

Let us state the main result of the section:
\begin{theorem}\label{thm1}
For a fixed $K\in\big(0,\frac12\pi\big)$, consider momenta $k>0$ such that $k\not \in \bigcup_{n\in\mathbb{N}_0}\left(\frac{n\pi-K}{\ell_2},\frac{n\pi+K}{\ell_2}\right)$. Suppose that the restriction of the generalized eigenfunction corresponding to energy $k^2$ to the elementary cell is normalized, then its components at the leftmost and rightmost vertical edges are at most of order $\mathcal{O}(k^{-1})$.
\end{theorem}
\begin{proof}
We focus on the vertex condition at the leftmost rightmost vertices. Since there are three edges meeting at each such junction, the vertex-scattering matrix expressed in term of the variable $\eta = \frac{1-k}{1+k}$ as in  \cite{ETa} behaves as
\begin{eqnarray}
  S(k) = \frac{1+\eta}{1+\eta+\eta^2}\begin{pmatrix}-\frac{\eta}{1+\eta} & 1 & \eta\\ \eta & -\frac{\eta}{1+\eta} & 1\\ 1 & \eta & -\frac{\eta}{1+\eta}\end{pmatrix}  = I + \mathcal{O}(k^{-1})\,.\label{eq:rec:s}
\end{eqnarray}
for $k\to\infty$ which is related to the fact that the matrix \eqref{matU} has an odd degree in this case, and therefore it does not have the eigenvalue $-1$, see also \cite[eq.~(8)]{KS}.

For definiteness, we consider the leftmost vertex with the aim to prove that the corresponding component of the generalized eigenfunction at this edge is small. We use the Ansatz $g_1(x) = a \,\mathrm{e}^{-ikx}+b\,\mathrm{e}^{ikx} = b\,\mathrm{e}^{ik\ell_2}\mathrm{e}^{-ik\hat x}+a\,\mathrm{e}^{-ik\ell_2}\mathrm{e}^{ik\hat x}$, where $\hat x := \ell_2 -x$ is the coordinate on the vertical edge with the orientation opposite to that of $x$. Similarly, we choose $f_1(x) = c \,\mathrm{e}^{-ikx}+d\,\mathrm{e}^{ikx}$. The scattering matrix at this vertex maps the vector of amplitudes of the incoming waves into the vector of the amplitudes of the outgoing waves, in other words
$$
  S(k) \begin{pmatrix}a\\b\,\mathrm{e}^{i(k\ell_2-\theta)}\\c\end{pmatrix} = \begin{pmatrix}b\\ a\,\mathrm{e}^{-i(k\ell_2+\theta)}\\d\end{pmatrix}\,.
$$
Using the explicit form \eqref{eq:rec:s} of $S(k)$ in combination with the previous equation the previous equation we can write
\begin{eqnarray}
  a-b & = & \mathcal{O}(k^{-1})b+\mathcal{O}(k^{-1})c\,,\label{eq:rec:abc1}\\
  b\,\mathrm{e}^{2ik\ell_2}-a & = & \mathcal{O}(k^{-1})a+ \mathcal{O}(k^{-1})c\,,\label{eq:rec:abc2}\\
  c-d & = & \mathcal{O}(k^{-1})a+\mathcal{O}(k^{-1})b+\mathcal{O}(k^{-1})c\,.\label{eq:rec:abc3}
\end{eqnarray}
Substituting for $a$ from \eqref{eq:rec:abc1} into \eqref{eq:rec:abc2} we obtain
\begin{equation} \label{estd}
  b\sin{k\ell_2} = \mathcal{O}(k^{-1})b+\mathcal{O}(k^{-1})c\,.
\end{equation}
Choosing now $k\ell_2 = n\pi +\delta$ with $n\in \mathbb{Z}$ and $\delta\in (-\pi/2,\pi/2]$ we get
$$
  b\big(\sin{\delta}+\mathcal{O}(k^{-1})\big)=\mathcal{O}(k^{-1})c\,.
$$
According to the assumption we fix a $K\in\big(0,\frac12\pi\big)0$ and consider $|\delta|>K>0$, then we have $b=\mathcal{O}(k^{-1})c$, thus \eqref{eq:rec:abc1} and \eqref{eq:rec:abc3} also yield $a=\mathcal{O}(k^{-1})c$ and $d=c+\mathcal{O}(k^{-1})c$, respectively. The first two relations in turn imply
$$
  \|g_1\|^2 = \int_0^{\ell_2} |a\,\mathrm{e}^{-ikx}+b\,\mathrm{e}^{ikx}|^2\,\mathrm{d}x \leq \int_0^{\ell_2}(|a|^2+|b|^2+2|a||b|)\,\mathrm{d}x = \mathcal{O}(k^{-2})|c|^2\,.
$$
The eigenfunction restriction to the period cell is supposed to be normalized, and therefore
\begin{multline*}
  1\geq \|f_1\|^2 = \int_0^{\ell_1} |c\,\mathrm{e}^{-ikx}+d\,\mathrm{e}^{ikx}|^2\,\mathrm{d}x = \int_0^{\ell_1} |2c\cos{kx}+\mathcal{O}(k^{-1})c|^2\,\mathrm{d}x =
\\
 = |c|^2\int_0^{\ell_1} [4\cos^2{kx}+\mathcal{O}(k^{-1})\cos{kx}+\mathcal{O}(k^{-2})]\,\mathrm{d}x = |c|^2 \big(2\ell_1+\mathcal{O}(k^{-1})\big)\,.
\end{multline*}
Hence we have $|c|^2\leq\frac{1}{2\ell_1}+\mathcal{O}(k^{-1})$, and consequently $\|g_1\|^2 \leq \mathcal{O}(k^{-2})$, in other words, $\|g_1\|\leq \mathcal{O}(k^{-1})$. The argument for the right strip edge is analogous.
\end{proof}

\section{The `brick' lattice}
\label{s:brick}

\begin{figure}
\centering
\includegraphics[height=8.5cm]{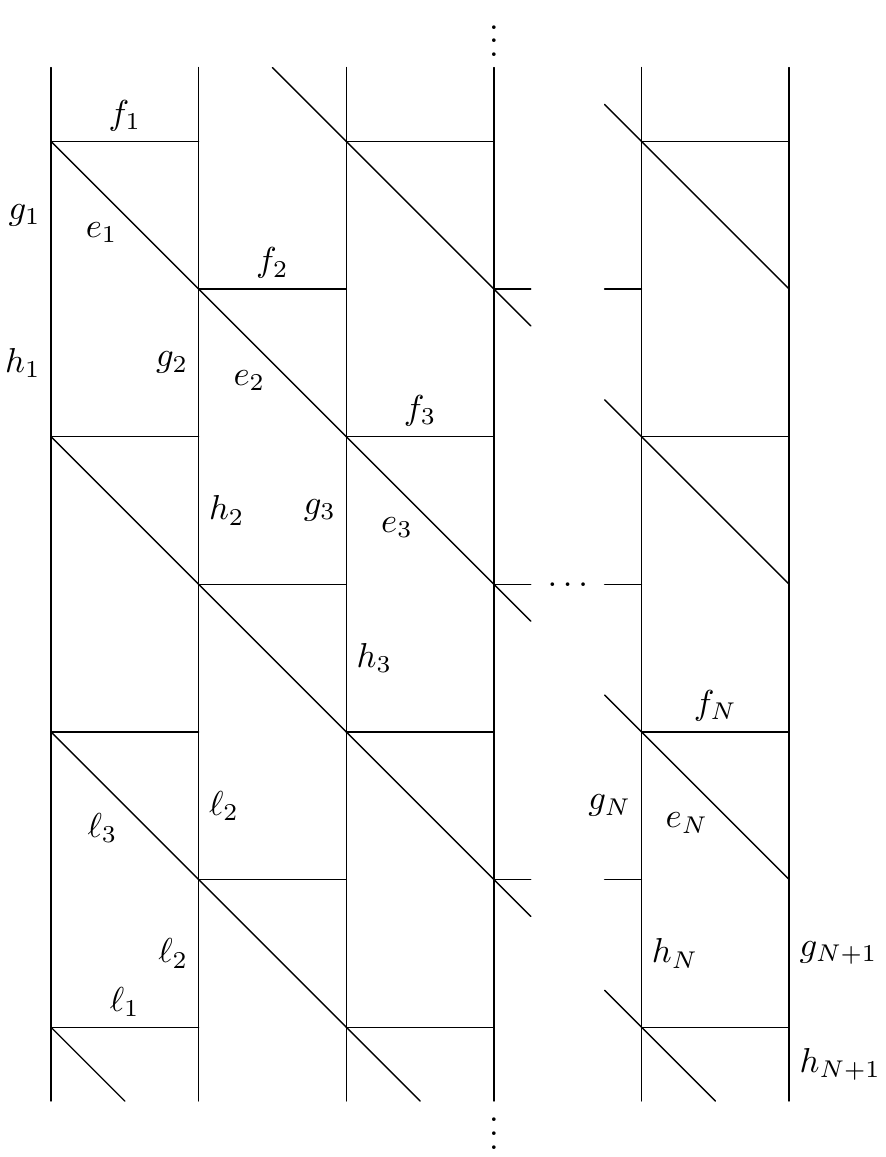}
\caption{``Brick'' lattice}
\label{fig2}
\end{figure}

Now we consider a lattice strip exhibiting the opposite effect having the transport related states localized (for most energy values) near a strip edge. The lattice structure is shown in Figure~\ref{fig2} and resembles a wall made of bricks with slanted lines. As in the previous case, we cut a vertical strip with $N$ cells in the horizontal direction. The horizontal edges are parametrized by the intervals $(0,\ell_1)$ with $x=0$ referring to the left vertex, the vertical edges by the intervals $(0,\ell_2)$ with $x=0$ at the top vertex (to keep the same notation, we describe the vertical edges on the left-hand side, with a `trivial' vertex in the middle, as composed of two edges of length $\ell_2$). Finally, the slanted edges are parametrized by the intervals $(0,\ell_3)$ with $x=0$ at the left-top vertex.

We denote the wavefunction components by $g_j$, $h_j$ (the vertical edges), $f_j$ (the horizontal edge) and $e_j$ (the slanted edge). The index $j$ denotes the cell and, for instance, the component $f_{j+1}$ is on the horizontal edge which is below and to the right from the edge supporting component $f_{j}$ (see Figure~\ref{fig2}). The `Bloch convention' is the same as in the previous example: the wavefunction components in the cells just above (below) the chosen periodicity cell are multiplied by the factor $\mathrm{e}^{\pm i\theta}$, respectively. We adopt the following Ansatz
\begin{eqnarray}
  g_j(x) = a_j \,\mathrm{e}^{-ikx}+b_j\,\mathrm{e}^{ikx}\,,&\quad &  h_j(x) = c_j \,\mathrm{e}^{-ikx}+d_j\,\mathrm{e}^{ikx}\,, \nonumber \\[-.75em] \label{Ansatz2} \\[-.75em]
  f_j(x) = p_j \,\mathrm{e}^{-ikx}+r_j\,\mathrm{e}^{ikx}\,,&\quad &  e_j(x) = s_j \,\mathrm{e}^{-ikx}+t_j\,\mathrm{e}^{ikx}\,.\nonumber
\end{eqnarray}
The key feature is that all the vertices of the graph except those on the most-left vertical line are of an odd degree, three or five. The matching conditions at these vertices look as follows: for $j = 1,\dots, N-1$ they are
$$
  \begin{pmatrix}-1 & 1 & 0 & 0 & 0\\ 0 &-1 & 1 & 0 & 0\\ 0 & 0 & -1 & 1 & 0\\0 & 0 & 0 & -1 & 1\\ 1 & 0 & 0 & 0 & -1\end{pmatrix}\begin{pmatrix} e_j(\ell_3)\\ g_{j+1}(0)\\ e_{j+1}(0) \\f_{j+1}(0)\\h_{j+1}(\ell_2) \mathrm{e}^{i\theta}\end{pmatrix}+ i\begin{pmatrix}1 & 1 & 0 & 0 & 0\\ 0 &1 & 1 & 0 & 0\\ 0 & 0 & 1 & 1 & 0\\0 & 0 & 0 & 1 & 1\\ 1 & 0 & 0 & 0 & 1\end{pmatrix}\begin{pmatrix} -e_j'(\ell_3)\\ g_{j+1}'(0)\\ e_{j+1}'(0) \\f_{j+1}'(0)\\-h_{j+1}'(\ell_2)\mathrm{e}^{i\theta} \end{pmatrix}= 0\,,\\
$$
for $j = 1,\dots, N$ we have
$$
  \begin{pmatrix}-1 & 1 & 0\\ 0 &-1 & 1\\ 1 & 0 & -1\end{pmatrix}\begin{pmatrix}f_j(\ell_1)\\h_{j+1}(0)\mathrm{e}^{i\theta}\\ g_{j+1}(\ell_2)\mathrm{e}^{i\theta}\end{pmatrix}+ i\begin{pmatrix}1 & 1 & 0\\ 0 &1 & 1\\ 1 & 0 & 1\end{pmatrix}\begin{pmatrix}-f_j'(\ell_1)\\h_{j+1}'(0)\mathrm{e}^{i\theta}\\ -g_{j+1}'(\ell_2)\mathrm{e}^{i\theta}\end{pmatrix} = 0
$$
and, finally,
$$
  \begin{pmatrix}-1 & 1 & 0\\ 0 &-1 & 1\\ 1 & 0 & -1\end{pmatrix}\begin{pmatrix} h_{N+1}(\ell_2) \mathrm{e}^{i\theta}\\e_N(\ell_3)\\ g_{N+1}(0)\end{pmatrix}+ i\begin{pmatrix}1 & 1 & 0\\ 0 &1 & 1\\ 1 & 0 & 1\end{pmatrix}\begin{pmatrix} -h'_{N+1}(\ell_2) \mathrm{e}^{i\theta}\\-e'_N(\ell_3)\\ g'_{N+1}(0)\end{pmatrix} = 0\,.
$$

To state the next result we denote by $q_j^{(m)}$, $m=1,\dots, 8$ subsequently the coefficients $a_j$, $b_j$, $c_j$, $d_j$, $p_j$, $r_j$, $s_j$ and $t_j$. Then we can make the following claim:
\begin{theorem} \label{thm2}
For a fixed $K\in\big(0,\frac12\pi\big)$, consider momenta $k>0$ such that
\begin{equation}\label{restr}
k\not \in \bigcup_{n\in\mathbb{N}_0}\left(\frac{n\pi-K}{\ell_1},\frac{n\pi+K}{\ell_1}\right)\cup \bigcup_{n\in\mathbb{N}_0}\left(\frac{n\pi-K}{\ell_2},\frac{n\pi+K}{\ell_2}\right) \cup \bigcup_{n\in\mathbb{N}_0}\left(\frac{n\pi-K}{\ell_3},\frac{n\pi+K}{\ell_3}\right).
\end{equation}
Suppose again that the restriction of the generalized eigenfunction corresponding to energy $k^2$ to the elementary cell is normalized, then $ q_j^{(m)}$ is at most of order $\mathcal{O}(k^{1-j})$ as $k\to\infty$.
\end{theorem}
\begin{proof}
As in the previous example, the scattering matrices at the vertices, except those lying on the leftmost vertical line, fulfill the relations
\begin{eqnarray}
  S_{j1}(k)\begin{pmatrix} t_j \mathrm{e}^{ik\ell_3}\\ a_{j+1}\\ s_{j+1}\\p_{j+1}\\d_{j+1} \mathrm{e}^{ik\ell_2}\mathrm{e}^{i\theta}\end{pmatrix} &= & \begin{pmatrix} s_j \mathrm{e}^{-ik\ell_3}\\ b_{j+1}\\ t_{j+1}\\r_{j+1}\\c_{j+1} \mathrm{e}^{-ik\ell_2}\mathrm{e}^{i\theta}\end{pmatrix}\,,\quad j = 1,\dots, N-1\,, \label{eq:brick:s1}\\
  S_{j2}(k)\begin{pmatrix}r_j\mathrm{e}^{ik\ell_1}\\c_{j+1}\mathrm{e}^{i\theta}\\b_{j+1}\mathrm{e}^{ik\ell_2}\mathrm{e}^{i\theta}\end{pmatrix} & = &\begin{pmatrix}p_j\mathrm{e}^{-ik\ell_1}\\d_{j+1}\mathrm{e}^{i\theta}\\a_{j+1}\mathrm{e}^{-ik\ell_2}\mathrm{e}^{i\theta}\end{pmatrix}\,,\quad j = 1,\dots, N\,,\label{eq:brick:s2}\\
  S_{N1}(k)\begin{pmatrix}d_{N+1} \mathrm{e}^{ik\ell_2}\mathrm{e}^{i\theta}\\t_N \mathrm{e}^{ik\ell_3}\\ a_{N+1}\end{pmatrix} & = &\begin{pmatrix}c_{N+1} \mathrm{e}^{-ik\ell_2}\mathrm{e}^{i\theta}\\s_N \mathrm{e}^{-ik\ell_3}\\ b_{N+1} \end{pmatrix}\,.\label{eq:brick:s3}
\end{eqnarray}
Using \eqref{eq:rec:s} and the analogous relation for vertices of degree five given in \cite{ETa} we infer that $S_{ji}(k) = I + \mathcal{O}(k^{-1})$ holds for $i=1,2$ and $j = 1,\dots N$ as $k\to\infty$, where $I$ denotes the identity matrix of the appropriate size. These relations together with eqs. \eqref{eq:brick:s1}--\eqref{eq:brick:s3} yield for $j = 1, \dots, N$ the identities
\begin{eqnarray}
t_j \mathrm{e}^{ik\ell_3}-s_j \mathrm{e}^{-ik\ell_3} &= &  \mathcal{O}(k^{-1})t_j+\mathcal{O}(k^{-1})a_{j+1}+\mathcal{O}(k^{-1})d_{j+1}+\mathcal{O}(k^{-1})p_{j+1}+\mathcal{O}(k^{-1})s_{j+1}\,,\label{eq:brick:eq1}\\
r_j \mathrm{e}^{ik\ell_1}-p_j \mathrm{e}^{-ik\ell_1} & = &  \mathcal{O}(k^{-1})r_j+ \mathcal{O}(k^{-1})b_{j+1}+ \mathcal{O}(k^{-1})c_{j+1}\,;\label{eq:brick:eq2}
\end{eqnarray}
for $j = 2, \dots, N+1$ we have
\begin{eqnarray}
a_j-b_j &= &  \mathcal{O}(k^{-1})t_{j-1}+ \mathcal{O}(k^{-1})a_j+ \mathcal{O}(k^{-1})d_j+ \mathcal{O}(k^{-1})p_j+ \mathcal{O}(k^{-1})s_j\,,\label{eq:brick:eq3}\\
b_j \mathrm{e}^{ik\ell_2}-a_j\mathrm{e}^{-ik\ell_2} & = & \mathcal{O}(k^{-1})r_{j-1}+ \mathcal{O}(k^{-1})b_j+ \mathcal{O}(k^{-1})c_j\,,\label{eq:brick:eq4}\\
c_j-d_j & = &  \mathcal{O}(k^{-1})r_{j-1}+ \mathcal{O}(k^{-1})b_j+ \mathcal{O}(k^{-1})c_j\,,\label{eq:brick:eq5}\\
d_j \mathrm{e}^{ik\ell_2}-c_j \mathrm{e}^{-ik\ell_2} &= &  \mathcal{O}(k^{-1})t_{j-1}+ \mathcal{O}(k^{-1})a_j+ \mathcal{O}(k^{-1})d_j+ \mathcal{O}(k^{-1})p_j+ \mathcal{O}(k^{-1})s_j\,,\label{eq:brick:eq6}
\end{eqnarray}
and finally, for $j = 2, \dots, N$,
\begin{eqnarray}
p_j -r_j &= &  \mathcal{O}(k^{-1})t_{j-1}+ \mathcal{O}(k^{-1})a_j+ \mathcal{O}(k^{-1})d_j+ \mathcal{O}(k^{-1})p_j+ \mathcal{O}(k^{-1})s_j\,,\label{eq:brick:eq7}\\
s_j - t_j &= &  \mathcal{O}(k^{-1})t_{j-1}+ \mathcal{O}(k^{-1})a_j+ \mathcal{O}(k^{-1})d_j+ \mathcal{O}(k^{-1})p_j+ \mathcal{O}(k^{-1})s_j\,.\label{eq:brick:eq8}
\end{eqnarray}
To take the right edge of the strip into account, we put here and in the following $p_{N+1}=r_{N+1} = s_{N+1} = t_{N+1} = 0$. Using \eqref{eq:brick:eq3} and \eqref{eq:brick:eq4} we find that for $j = 2,\dots, N+1$ it holds
\begin{eqnarray}
b_j \sin{k\ell_2} =  \mathcal{O}(k^{-1})t_{j-1}+ \mathcal{O}(k^{-1})r_{j-1}+ \mathcal{O}(k^{-1})b_{j}+ \mathcal{O}(k^{-1})d_j+ \mathcal{O}(k^{-1})p_j+ \mathcal{O}(k^{-1})s_j\,,\label{eq:brick:eq9}\\
a_j   =  b_j + \mathcal{O}(k^{-1})t_{j-1}+ \mathcal{O}(k^{-1})r_{j-1}+ \mathcal{O}(k^{-1})b_{j}+ \mathcal{O}(k^{-1})d_j+ \mathcal{O}(k^{-1})p_j+ \mathcal{O}(k^{-1})s_j\,,\label{eq:brick:eq10}
\end{eqnarray}
and similarly, \eqref{eq:brick:eq5} and \eqref{eq:brick:eq6} yield for $j = 2,\dots, N+1$
\begin{eqnarray}
d_j \sin{k\ell_2} =  \mathcal{O}(k^{-1})t_{j-1}+ \mathcal{O}(k^{-1})r_{j-1}+ \mathcal{O}(k^{-1})b_{j}+ \mathcal{O}(k^{-1})d_j+ \mathcal{O}(k^{-1})p_j+ \mathcal{O}(k^{-1})s_j\,,\label{eq:brick:eq11}\\
c_j   =  d_j + \mathcal{O}(k^{-1})t_{j-1}+ \mathcal{O}(k^{-1})r_{j-1}+ \mathcal{O}(k^{-1})b_{j}+ \mathcal{O}(k^{-1})d_j+ \mathcal{O}(k^{-1})p_j+ \mathcal{O}(k^{-1})s_j\,.\label{eq:brick:eq12}
\end{eqnarray}
Furthermore, from \eqref{eq:brick:eq1} and \eqref{eq:brick:eq8} we get for $j = 2,\dots, N$
\begin{multline}
t_j \sin{k\ell_3} =  \mathcal{O}(k^{-1})t_{j-1}+\mathcal{O}(k^{-1})t_{j}+ \mathcal{O}(k^{-1})a_{j}+ \mathcal{O}(k^{-1})a_{j+1}+ \mathcal{O}(k^{-1})d_j+ \mathcal{O}(k^{-1})d_{j+1}\\
+\mathcal{O}(k^{-1})p_{j}+ \mathcal{O}(k^{-1})p_{j+1}+ \mathcal{O}(k^{-1})s_{j}+ \mathcal{O}(k^{-1})s_{j+1}\,,\label{eq:brick:eq13}
\end{multline}
\begin{equation}
s_j   =  t_j + \mathcal{O}(k^{-1})t_{j-1}+ \mathcal{O}(k^{-1})a_j+ \mathcal{O}(k^{-1})d_j+ \mathcal{O}(k^{-1})p_j+ \mathcal{O}(k^{-1})s_j\,,\label{eq:brick:eq14}
\end{equation}
and finally, \eqref{eq:brick:eq2} and \eqref{eq:brick:eq7} imply for $j = 2,\dots, N$
\begin{eqnarray}
\hspace{-5mm}r_j \sin{k\ell_1} &=& \mathcal{O}(k^{-1})t_{j-1}+ \mathcal{O}(k^{-1})a_{j}+ \mathcal{O}(k^{-1})d_j+ \mathcal{O}(k^{-1})s_{j}+\mathcal{O}(k^{-1})p_{j}+ \mathcal{O}(k^{-1})r_{j} + \mathcal{O}(k^{-1})b_{j+1}+ \mathcal{O}(k^{-1})c_{j+1}\,,\label{eq:brick:eq15}\\
p_j   &=&  r_j + \mathcal{O}(k^{-1})t_{j-1}+ \mathcal{O}(k^{-1})a_{j}+ \mathcal{O}(k^{-1})d_j+ \mathcal{O}(k^{-1})p_j+ \mathcal{O}(k^{-1})s_j\,.\label{eq:brick:eq16}
\end{eqnarray}
Now we can argue in analogy with the proof of Thm.~\ref{thm1} and summarize the obtained relations for the momentum values satisfying condition \eqref{restr} as
\begin{equation}
  q_j^{(m)} = \mathcal{O}(k^{-1}) \left(|r_{j-1}|+|t_{j-1}|+ \sum_{s=1}^{8}|q_j^{(s)}|+\sum_{s=1}^{8}|q_{j+1}^{(s)}| \right)\,,\label{eq:brick:eq:q1}
\end{equation}
The previous equation holds for $j=2,\dots, N$, $\:m=1,\dots, 8$, and the last sum for $j=N$ corresponding to the right strip edge includes $|a_{N+1}|$, $|b_{N+1}|$, $|c_{N+1}|$, and $|d_{N+1}|$ only . Moreover, for $m = 1,\dots, 4$ we have
\begin{equation}
  q_{N+1}^{(m)} = \mathcal{O}(k^{-1}) \left(|r_{N}|+|t_{N}|+ \sum_{s=1}^{4}|q_{N+1}^{(s)}|\right)\,.\label{eq:brick:eq:q2}
\end{equation}
Eliminating successively the variables, we can conclude from eqs. \eqref{eq:brick:eq:q1} and \eqref{eq:brick:eq:q2} that for $j=2,\dots, N+1$ one has
$$
  q_j^{(m)} = \mathcal{O}(k^{-1})(|r_{j-1}|+|t_{j-1}|) = \mathcal{O}(k^{1-j})(|r_{1}|+|t_{1}|)\,;
$$
from this relation in combination with the fact that $\|e_j\|+\|f_j\|\le 1$ holds by assumption the claim follows in the same way as in the proof of Thm.~\ref{thm1}.
\end{proof}

Let us conclude this section with a few remarks:

\medskip

(a) In contrast to Theorem~\ref{thm1} the edges here are not equivalent, it is the left one which is conducting. The claim of Theorem~\ref{thm2} can be easily extended to the situation where the strip is replaced by a halfplane.

\smallskip

(b) Note that the positive constant $K$ in both theorems can be chosen arbitrarily small, however, the estimate used in \eqref{estd} and its analogue in the proof of Theorem~\ref{thm2} show that the observed coefficient decay becomes worse as the momentum approaches the values\footnote{In the quantum graph jargon, the are sometimes referred to as `Dirichlet points', in the present context `Neumann' would be more appropriate having the effective edge decoupling at high energies in mind.} $\,k=\pi/\ell_j$.

\smallskip

(c) Note also that the scattering matrices referring to vertices at the strip edges play no role in the proofs. For vertices of degree four with our coupling all the transmition  and reflection probabilities are asymptotically the same \cite{ETa}, hence one expects the quantity $\frac{|a_1|^2-|b_1|^2}{|a_1|^2+|b_1|^2}$ (and similarly for $c_1,\,d_1$) describing the (relative) probability current along the left strip edge to range approximately between 1 and -1 as a function of $\theta$ for $k$ large enough, cf. the end of Section~\ref{ss:brick} below.

\section{Numerical results}

To get a better idea about the described effects and to show they are pronounced already for small values of the strip width $N$, let us add numerical illustrations. For both types of lattice strips, we find the coefficients of the generalized eigenfunctions for certain points in the absolutely continuous spectrum and show that for $k$ far from integer multiples of $\pi/\ell_j$, the coefficients decay in accordance with the general results obtained above.

\subsection{Rectangular lattice}

\begin{figure}
\centering
\includegraphics[height=10cm]{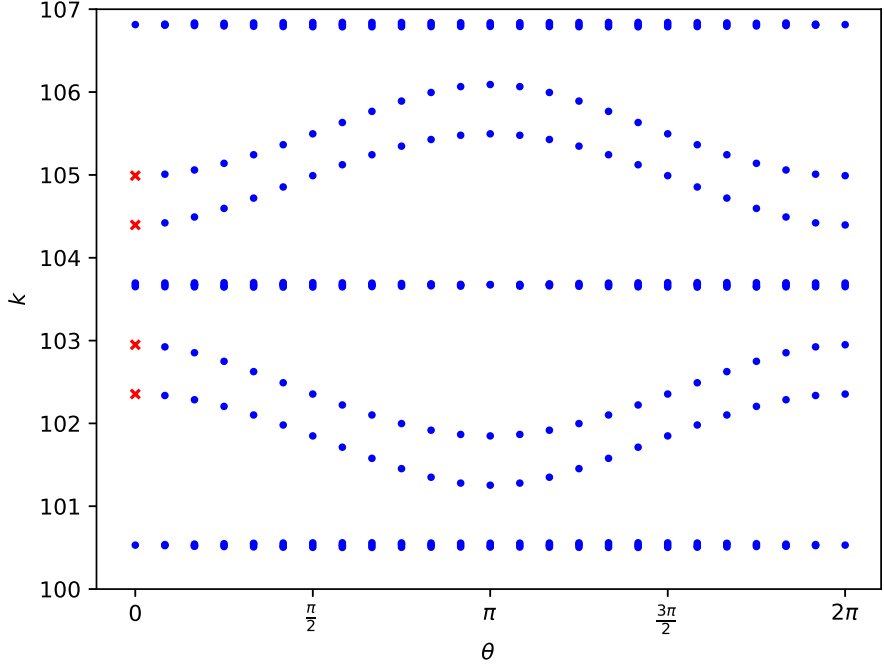}
\caption{Dispersion diagram for the rectangular lattice strip, the points marked by red crossed are used in Figure~\ref{fig4} (color online)}
\label{fig3}
\end{figure}

Consider the strip with $N=3$ and $\ell_1 = \ell_2 = 1$. The spectrum is found in the standard way, using the wavefunction Ansatz together with the matching conditions at the vertices we get a linear equation system for the coefficients $a,b$; its solvability condition for a fixed quasimomentum $\theta$ yields the eigenvalues of the corresponding Bloch component of the Hamiltonian. We did that for discrete values of $\theta$ with the step $\pi/12$ and the momentum $k$ higher than 100; the resulting dispersion diagram is shown in Figure~\ref{fig3}. We normalize the vector of coefficients by $\sin{(kx)}$ and $\cos{(kx)}$ in the description of the eigenfunction components on a single period cell.

It is obvious that the dispersion curves are of two sorts. Some of them, here in the vicinity of $32\pi \approx 100.531$, $33\pi \approx 103.673$, and $34\pi \approx 106.814$ correspond to narrow, almost flat bands. In contrast, the other bands are wide and fulfill well the assumptions of Theorem~\ref{thm1}; we choose their edge points marked by red crosses, corresponding to $\theta = 0$ and $k\approx 102.354, 102.949, 104.396, 104.991$, to plot the eigenfunction coefficients in order to illustrate how are they suppressed at the strip boundaries.

Figure~\ref{fig4} shows the values of two combinations of the coefficients $a_j,\,b_j$ (four points connected with a line) and $c_j,\,d_j$ (three points connected with a line). In the left panel, we plot $||b_j|^2-|a_j|^2|$, proportional to the absolute value of the probability current through the given edge, and a similar expression for coefficients $c_j,\,d_j$. The right panel shows $|a_j|^2+|b_j|^2$ and $|c_j|^2+|d_j|^2$. The $y$-axis scale is logarithmic and one can notice that the values of the plotted expressions on the leftmost and rightmost vertical edges are three to four orders of magnitude smaller than those the middle edges. This correspondents to the reduction of the function at the strip edges by a factor of approximately $10^{-2}$ in accordance with Theorem~\ref{thm1}. One can also notice that there is a certain, tiny indeed probability current on the vertical edges, that is, across the strip, while the wavefunction values themselves at these edges are not small.

\begin{figure}
\centering
\begin{subfigure}{.5\textwidth}
  \centering\captionsetup{width=.9\linewidth}
  \includegraphics[width=.99\linewidth]{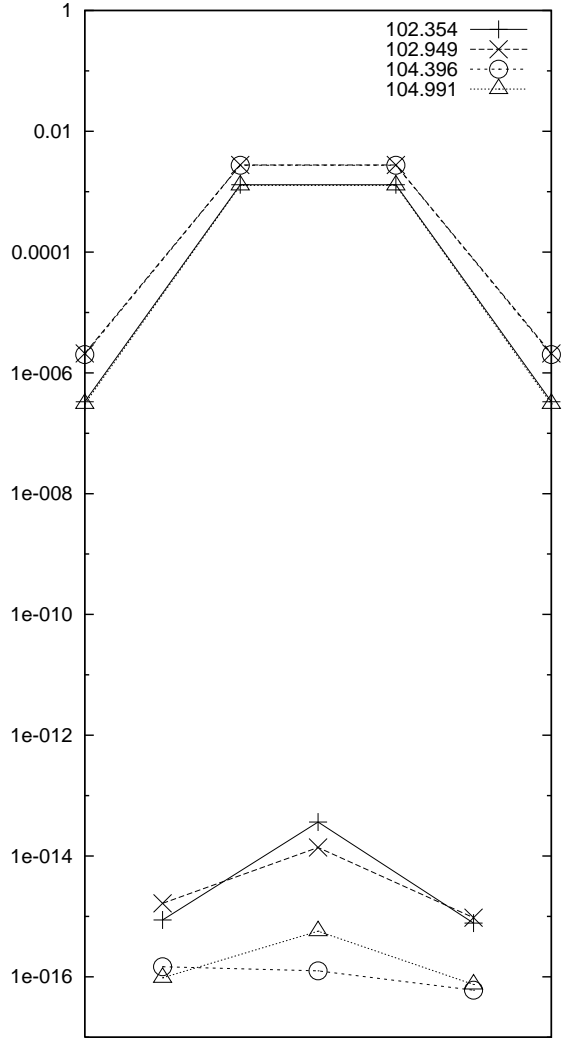}
  \caption{Values of coeficients $||b_j|^2-|a_j|^2|$ and $||d_j|^2-|c_j|^2|$, proportional to probability current}
  \label{fig4a}
\end{subfigure}%
\begin{subfigure}{.5\textwidth}
  \centering\captionsetup{width=.9\linewidth}
  \includegraphics[width=.99\linewidth]{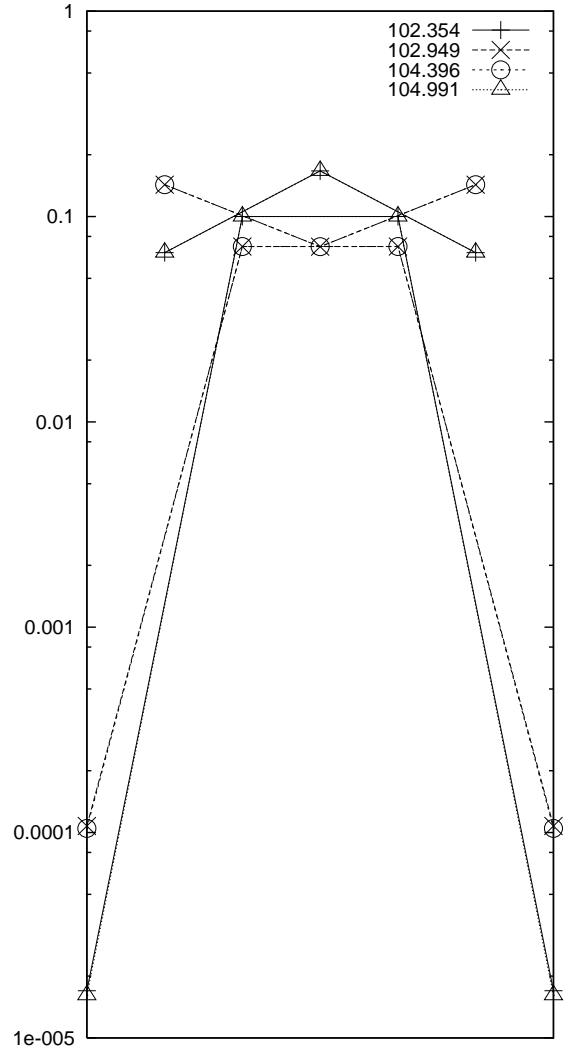}
  \caption{Values of coeficients $|a_j|^2+|b_j|^2$ and $|c_j|^2+|d_j|^2$}
  \label{fig4b}
\end{subfigure}
\caption{Rectangular lattice strip $N=3$: values of two combinations of the coefficients $a_j$, $b_j$, $c_j$, $d_j$. For the particular values of $k$ in the absolutely continuous spectrum shown in the caption (102.354, 102.949, 104.396, and 104.991) and $\theta = 0$ we plot the modulus of the difference $|b_j|^2-|a_j|^2$ (or $|d_j|^2-|c_j|^2$) proportional to the probability current (left panel) and the sum $|a_j|^2+|b_j|^2$ (or $|c_j|^2+|d_j|^2$) (right panel). The values are connected by lines indicating the value of $k$. The points are arranged according to the corresponding vertex positions in the strip, the three connected points correspond to the coefficients on the horizontal edges. }
\label{fig4}
\end{figure}

\subsection{The `brick' lattice}
\label{ss:brick}

In this case we choose an even smaller strip width, $N=2$, and put $\ell_1 = \ell_2 = 1$, $\ell_3 = \sqrt{2}$. Using \eqref{Ansatz2} we find the dispersion curves shown on the diagram in Figure~\ref{fig5} for the window ranging from $k = 100$ to $k =107$. As in the previous example, we ignore the almost flat bands that appear in the vicinity of $\frac{46\pi}{\sqrt{2}} \approx 102.186$, $\frac{47\pi}{\sqrt{2}}$, $32\pi \approx 100.531$, $33\pi\approx 103.673$, and $34\pi\approx 106.814$, and focus on the wide ones, specifically at the points corresponding $\theta = 0$ and $k\approx 101.133, 103.103, 104.046, 104.677$ marked by red crosses and plot the same coefficient combinations as before.

Figure~\ref{fig6} shows in the left panel $||b_j|^2-|a_j|^2|$ and the analogous expressions for the other edges and in the right panel $|a_j|^2+|b_j|^2$ and its analogues. The $x$-coordinate of the points from the left to right corresponds to the values of these expressions on the edges $g_1$, $h_1$, $f_1$, $e_1$, $g_2$, $h_2$, $f_2$, $e_2$, $g_3$, $h_3$. We see that the indicated expressions with indices 2 are three to four orders of magnitude smaller than the largest expressions (which are in the right panel). Similarly, the expressions with index 3 are six to eight orders smaller. Accordingly, the components of the normalized generalized eigenfunctions are roughly reduced by factor  $10^{-2}$ for the edges with index 2 and $10^{-4}$ for the edges with index 3, in agreement with Theorem~\ref{thm2}.

Note finally that the probability current characterized by the quantity $\frac{|a_1|^2-|b_1|^2}{|a_1|^2+|b_1|^2}$, cf. Remark~(c) above, is suppressed at the band edges, $\theta=0$, where the value is $\approx -0.02$ but it ranges roughly between 1 and -1 as $\theta$ runs through the Brillouin zone, cf.~Figure~\ref{fig7}. The figure is plotted for the band starting with $k\approx 101.133$, the behavior for the other bands is similar.

\begin{figure}
\centering
\includegraphics[height=10cm]{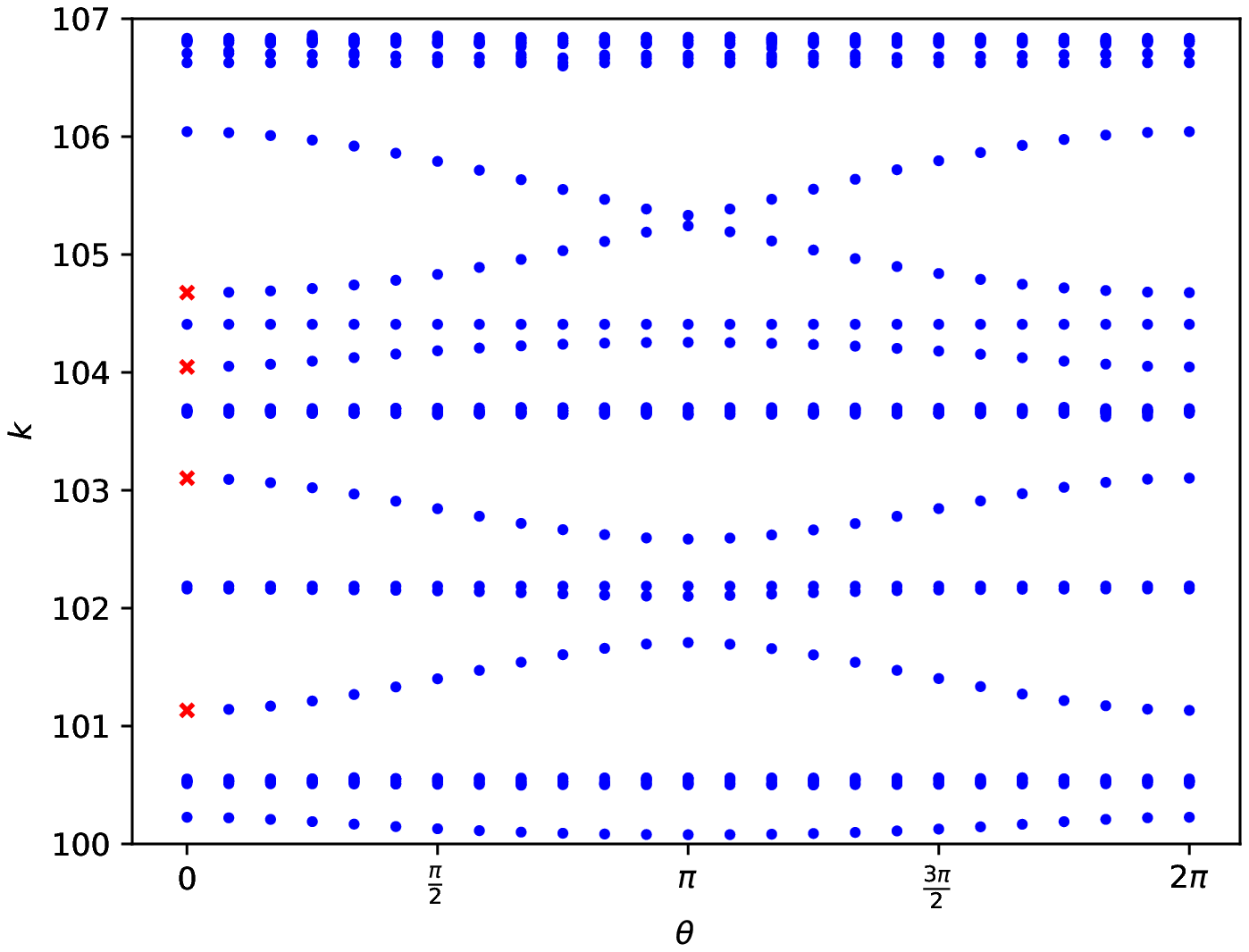}
\caption{Dispersion diagram for the `brick'' lattice strip, the points marked by red crossed are used in Figure~\ref{fig6} (color online)}
\label{fig5}
\end{figure}

\begin{figure}
\centering
\begin{subfigure}{.5\textwidth}
  \centering\captionsetup{width=.9\linewidth}
  \includegraphics[width=.99\linewidth]{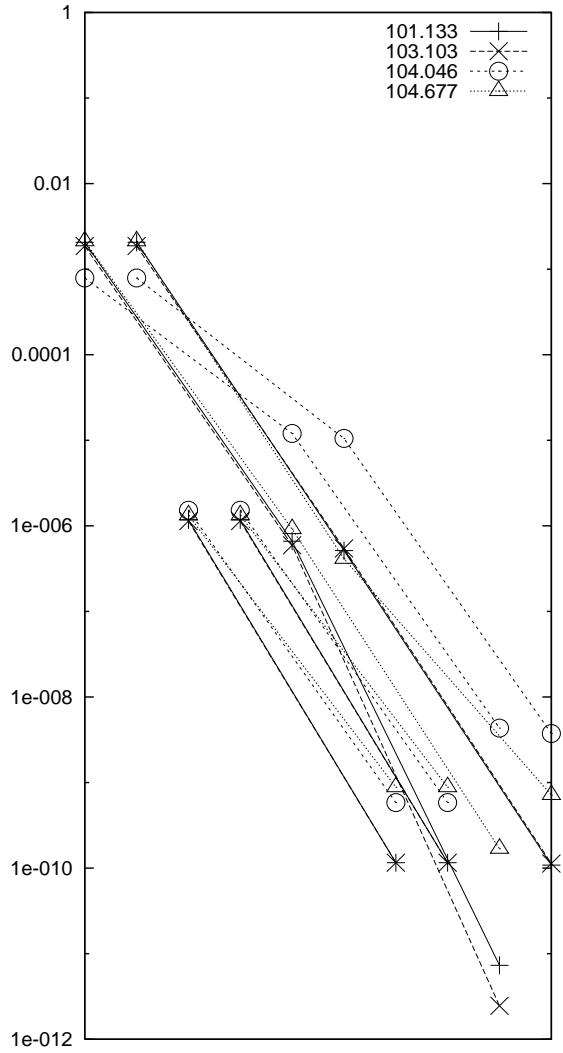}
  \caption{Values of coeficients $||b_j|^2-|a_j|^2|$, $||d_j|^2-|c_j|^2|$, $||r_j|^2-|p_j|^2|$, and $||t_j|^2-|s_j|^2|$, proportional to probability current}
  \label{fig6a}
\end{subfigure}%
\begin{subfigure}{.5\textwidth}
  \centering\captionsetup{width=.9\linewidth}
  \includegraphics[width=.99\linewidth]{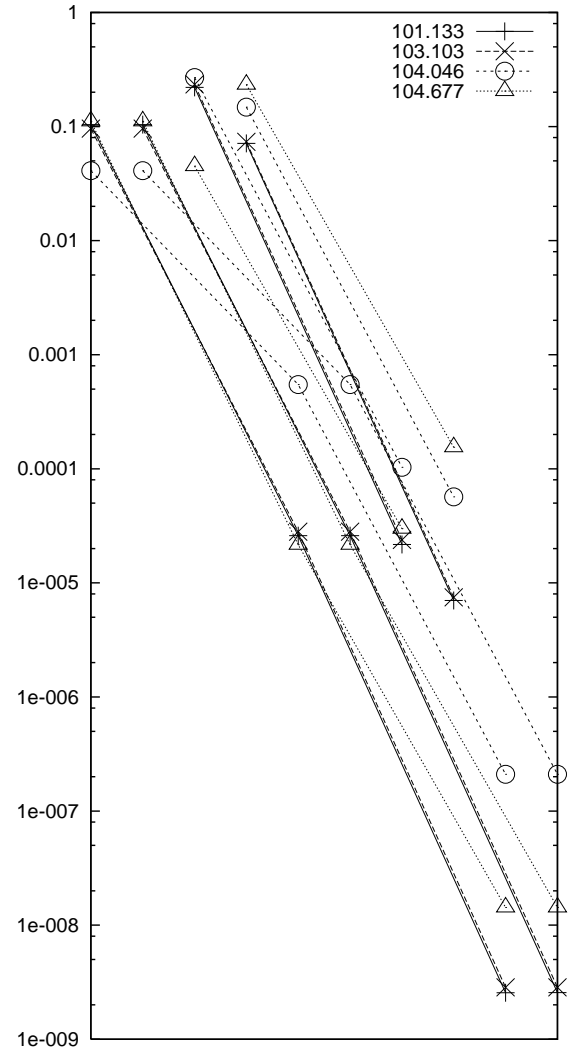}
  \caption{Values of coeficients $|a_j|^2+|b_j|^2$, $|c_j|^2+|d_j|^2$, $|p_j|^2+|r_j|^2$, $|s_j|^2+|t_j|^2$}
  \label{fig6b}
\end{subfigure}
\caption{`Brick' lattice strip with $N=2$: combinations of the coefficients $a_j$, $b_j$, $c_j$, $d_j$, $p_j$, $r_j$, $s_j$, and $t_j$. For the values of $k$ indicated in the caption (101.133, 103.103, 104.046, and 104.677) and $\theta = 0$ we plot the modulus of the difference $|b_j|^2-|a_j|^2$, etc. (left panel) and the sum $|a_j|^2+|b_j|^2$, etc. (right panel). The expressions made of $a_j$ and $b_j$ are connected by lines (similarly for the other coefficients). The leftmost point refers to the function $g_1$, it is connected to the to those of the other $g_j$. The second point from the left corresponds to function $h_1$  being connected to the other edges of the same type. The next point from the left corresponds to $f_1$, the other to $e_1$, etc.}
\label{fig6}
\end{figure}

\begin{figure}
\centering
\includegraphics[height=8.5cm]{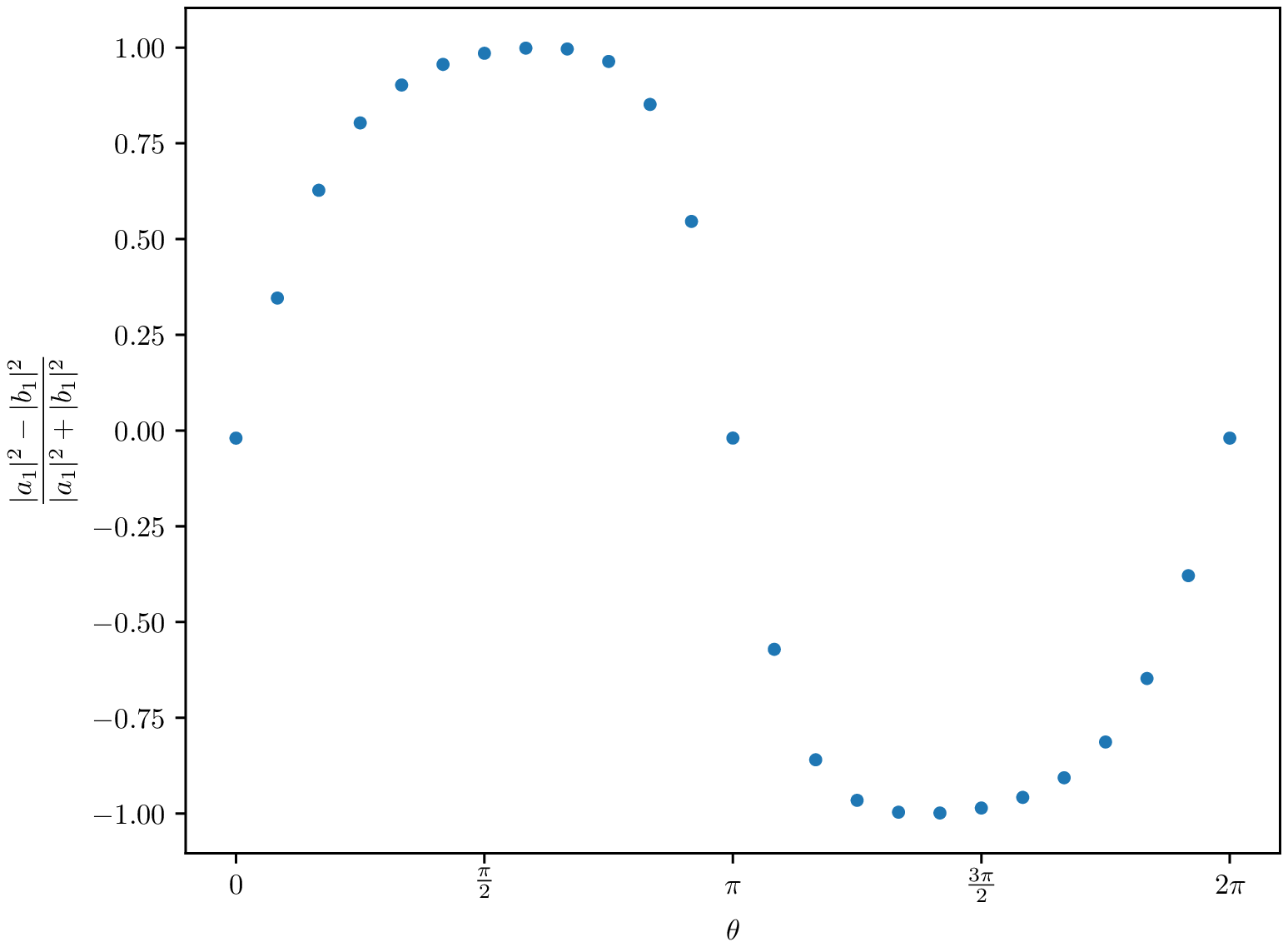}
\caption{The (relative) probability current along the left strip edge as a function of $\theta$ for the band starting with $k\approx 101.133$}
\label{fig7}
\end{figure}

\section*{Acknowledgements}
P.E. acknowledges support of the EU project CZ.02.1.01/0.0/0.0/16\textunderscore 019/0000778, and J.L. of the research programme ``Mathematical Physics and Differential Geometry'' of the Faculty of Science of the University of Hradec Kr\'alov\'e.

\end{document}